\documentclass[11pt]{article}
\usepackage[T1]{fontenc}
\usepackage{graphicx}

\usepackage{fullpage}

\usepackage[utf8]{inputenc}
\usepackage{nicefrac}
\usepackage{tikz}
\usetikzlibrary{decorations.pathreplacing,calligraphy}
\usepackage{amsmath}
\usepackage{amsthm}
\usepackage{amsfonts}
\usepackage{multicol}
\usepackage{booktabs}   %
\usepackage[linesnumbered,lined,boxed,commentsnumbered]{algorithm2e}

\newtheorem{theorem}{Theorem}
\newtheorem{proposition}{Proposition}
\newtheorem{definition}{Definition}
\newtheorem{corollary}{Corollary}

\newcommand{\np}{{{\mathrm{NP}}}}
\newcommand{\p}{{{\mathrm{P}}}}
\newcommand{\pref}{\succ}

\newcommand{\calR}{\mathcal{R}}
\newcommand{\calS}{\mathcal{S}}
\newcommand{\calT}{\mathcal{T}}

\newcommand{\score}{{{\mathrm{score}}}}

\newcommand{\cc}{{{\mathrm{CC}}}}
\newcommand{\pav}{{{\mathrm{PAV}}}}

\begin{document}
\title{Robustness of Greedy Approval Rules\thanks{See
    \texttt{https://github.com/Project-PRAGMA/Greedy-Robust-EUMAS-2022}
    for the source code of the experiments performed in this paper.}}

\author{Piotr Faliszewski\\
  AGH University\\
  Krakow, Poland
  \and
  Grzegorz Gawron\\
  AGH University and\\ VirtusLab\\
  Krakow, Poland
  \and
  Bartosz Kusek\\
  AGH University\\
  Krakow, Poland}

\maketitle
\begin{abstract}
  We study the robustness of GreedyCC, GreedyPAV, and Phargm{\'e}n's
  sequential rule, using the framework introduced by Bredereck et
  al.~\cite{bre-fal-kac-nie-sko-tal:c:robustness} for the case of
  (multiwinner) ordinal elections and adopted to the approval setting
  by Gawron and
  Faliszewski~\cite{gaw-fal:c:approval-robustness}. First, we show
  that for each of our rules and every committee size $k$, there are
  elections in which adding or removing a certain approval causes the
  winning committee to completely change (i.e., the winning committee
  after the operation is disjoint from the one before the
  operation). Second, we show that the problem of deciding how many
  approvals need to be added (or removed) from an election to change
  its outcome is $\np$-complete for each of our rules. Finally, we
  experimentally evaluate the robustness of our rules in the presence
  of random noise.
\end{abstract}

\section{Introduction}
We study the extent to which perturbing the input of several
approval-based multiwinnner voting rules affects their outcome.  We
focus on GreedyCC, GreedyPAV, and Phragm{\'e}n rules, whose common
feature is that they choose members of the winning committee in a
sequential, greedy way.

In a multiwinner approval election, each voter indicates which
candidates he or she finds appealing---i.e., which ones he or she
approves---and a voting rule provides the winning committee (i.e., a
fixed-size group of candidates).
For example, the approval voting rule (AV) chooses committees of
individually excellent candidates (i.e., the most approved ones), the
proportional approval voting rule (PAV) ensures proportional
representation of the voters, and the Chamberlin-Courant rule (CC)
seeks a diverse committee.  Unfortunately, while AV can be computed in
polynomial time, finding the winning committees under the other two
rules is intractable. Luckily, there are many workarounds for this
issue. For example, instead of CC we can use its approximate variant
GreedyCC, and instead of PAV we can either use GreedyPAV or the
Phragm{\'e}n rule (see the oveview of Lackner and
Skowron~\cite{lac-sko:t:approva-multiwinner-survey} for a discussion
of these rules and their properties).
While there is robustness analysis of AV, CC, and
PAV~\cite{gaw-fal:c:approval-robustness}, analogous results are
missing for these greedy rules and our goal is to fill this hole.

We use the robustness framework of Bredereck et
al.~\cite{bre-fal-kac-nie-sko-tal:c:robustness}, as adopted to the
case of approval elections by Gawron and
Faliszewski~\cite{gaw-fal:c:approval-robustness}. This framework
consists of the following elements:
\begin{enumerate}
\item Evaluating the extent to which introducing a single small change
  may affect the outcome of a rule. For example, we say that a rule
  has \textsc{Add}-robustness level equal to $\ell$ if adding a single
  approval results in, at most, replacing~$\ell$ committee
  members. \textsc{Remove}-robustness level is defined analogously,
  but for the case of removing a single approval. Robustness levels of
  a rule describe its worst-case behavior under minimal perturbations
  of the input.\footnote{Whenever we speak of ``robustness levels''
    without indicating whether we mean the \textsc{Add} or
    \textsc{Remove} variant, we collectively refer to both.}

\item Establishing the complexity of the \textsc{Robustness-Radius}
  problem, which asks if a given number of basic operations (such as
  adding or removing approvals) suffices to change the election
  outcome. Solving this problem for various elections would measure a
  rule's robustness to targeted attacks on a per-instance basis. However,
  since \textsc{Robustness-Radius} is $\np$-complete for many rules,
  neither Bredereck et al.~\cite{bre-fal-kac-nie-sko-tal:c:robustness}
  nor Gawron and Faliszewski~\cite{gaw-fal:c:approval-robustness}
  carried out such experiments and we follow them in this respect.

\item Computing for various elections how many randomly selected basic
  operations are needed, on average, to change their outcomes.  This
  measures the rules' robustness to random noise.
\end{enumerate}
Gawron and Faliszewski~\cite{gaw-fal:c:approval-robustness} considered
AV, SAV (a rule similar in spirit to AV), CC, and PAV.  They have
shown that AV has $\{$\textsc{Add}, \textsc{Remove}$\}$-robustness
levels equal to $1$, while the other rules have them equal to the
committee size (although there are some intricacies for the case of
SAV). Further, they have shown that
\textsc{Robustness-Radius} is in $\p$ for AV and SAV, but is
$\np$-hard for CC and PAV. Given hardness of computing CC and PAV,
this last result is not very surprising, but Gawron and Faliszewski
have also shown fixed-parameter tractable algorithms for the
respective problems.  Unfortunately, Gawron and
Faliszewki~\cite{gaw-fal:c:approval-robustness} did not pursue
experimental studies (as some of their rules are $\np$-hard, even
computing robustness to random noise would require nontrivial
computing resources).

\paragraph{Our Contribution.}
We complement the work of Gawron and
Faliszewski~\cite{gaw-fal:c:approval-robustness} by considering
GreedyCC, GreedyPAV, and Phragm{\'e}n. We show that their
$\{$\textsc{Add}, \textsc{Remove}$\}$-robustness levels are equal to
the committee size and we show that the \textsc{Robustenss-Radius}
problem is $\np$-complete for each of them. Since our rules are
polynomial-time to compute, this result is not as immediate as in the
case of CC or PAV.  Finally, we experimentally evaluate the robustness
of our rules, and of AV, to random noise.

\paragraph{Related Work.}

In addition to the works of Bredereck et
al.~\cite{bre-fal-kac-nie-sko-tal:c:robustness} and Gawron and
Faliszewski~\cite{gaw-fal:c:approval-robustness}, our results are
closely related to the line of work on the complexity of bribery in
elections.  In a bribery problem, we are given an election and we ask
if a certain outcome---such as including a certain candidate among the
winners (in the constructive setting) or precluding a certain
candidate from winning (in the destructive setting)---can be achieved
by modifying the preferences of the voters with operations of a
certain cost.  The study of bribery was initiated by Faliszewski,
Hemaspaandra, and Hemaspaandra~\cite{fal-hem-hem:j:bribery} and was
continued by many others (see the overview of Faliszewski and
Rothe~\cite{fal-rot:b:control-bribery}).  The
\textsc{Robustness-Radius} problem can be seen as a variant of
destructive bribery. \textsc{Swap-Bribery}, introduced by Elkind,
Faliszewski, and Slinko~\cite{elk-fal-sli:c:swap-bribery}, was used to
study the robustness of single-winner voting rules by Shiryaev et
al.~\cite{shi-yu-elk:c:robust-winners}, Baumeister and
Hogrebe~\cite{bau-hog:c:robustness}, and Boehmer et
al.~\cite{boe-bre-fal-nie:c:counting-swap-bribery}.  Magrino et
al.~\cite{mag-riv-she-wag:c:stv-bribery},
Cary~\cite{car:c:stv-bribery}, and Xia~\cite{xia:c:margin-of-victory}
used variants of destructive bribery to study margin of victory under
various single-winner voting rules.  The main difference between the
studies of robustness and margin of victory is that in the former, the
authors typically use fine-grained bribery variants that allow for
making small modifications in the votes (in our case, these mean
adding or removing single approvals), whereas in the latter the
authors typically use coarse-grained bribery variants that allow
operations that change the whole votes arbitrarily.

Our work is closely related to that of Faliszewski, Skowron and
Talmon~\cite{fal-sko-tal:c:bribery-success}, who study bribery of
approval-based multiwinner rules under adding, removing, and swapping
approvals. The main difference between our work and theirs is that
they focus on the constructive setting and we study the
destructive one.

\section{Preliminaries}

We write $\mathbb{N}_+$ to denote the set $\{1,2, \ldots\}$ and for each
integer $t$, by $[t]$ we mean the set $\{1, \ldots, t\}$. We use the
Iverson bracket notation, i.e., given a logical expression~$P$, we write
$[P]$ to mean $1$ if $P$ is true and to mean $0$ otherwise.

\subsection{Approval Elections and Multiwinner Rules}
An election is a pair $E = (C,V)$, where $C = \{c_1, \ldots, c_m\}$ is
a set of candidates and $V = (v_1, \ldots, v_n)$ is a collection of
voters. Each voter $v_i$ has a set $A(v_i) \subseteq C$ of candidates
that he or she approves. The approval score of a candidate is the
number of voters that approve him or her.

A multiwinner voting rule $\calR$ is a function that given an election
$E$ and committee size $k$ outputs a family of size-$k$ winning
committees (i.e., a family of size-$k$ subsets of $C$).  If a rule
always outputs a unique committee, then we say that it is resolute (in
practice, non-resolute rules require tie-breaking rules, but we
disregard this issue). For example, the approval voting rule (the AV
rule) selects committees of $k$ candidates with the highest approval
scores.  AV belongs to the class of Thiele rules, which are defined as
follows: Consider an election $E = (C,V)$ and a nonincreasing function
$\lambda \colon \mathbb{N}_+ \rightarrow [0,1]$, such that
$\lambda(1) = 1$ (we will refer to functions satisfying these
conditions as OWA functions\footnote{The name refers to the class of
  order-weighted operators (OWA operators), introduced by
  Yager~\cite{yag:j:owa} and used by Skowron et
  al.~\cite{sko-fal-lan:j:collective} to define a class of rules
  closely related to the Thiele ones}).  We define the $\lambda$-score
of a set $S \subseteq C$ of candidates as:
\[
  \textstyle \lambda\hbox{-}\score_E(S) = \sum_{v \in V} \left(
    \sum_{t=1}^{|S \cap A(v)|} \lambda(t) \right).
\]
Given an election $E$ and committee size $k$, the $\lambda$-Thiele
rule outputs those size-$k$ committees $W$ that have the highest
$\lambda$-score. For example, the AV rule uses
the constant  function $\lambda_{AV}(i) = 1$. 
This rule is meant to choose committees of individually excellent
candidates, hence it considers the candidates with the highest individual
approval scores.
We are also interested in the
Chamberlin--Courant rule (the CC rule) and the proportional approval
voting rule (the PAV rule), which use functions
$\lambda_{\cc}(i) = [i = 1]$ and
$\lambda_{\pav}(i) = \nicefrac{1}{i}$, respectively.
Under CC, a voter assigns score $1$ to a committee exactly if he or she
approves at least one of its members, and assigns score $0$ otherwise.
This rule was introduced by Chamberlin and Courant~\cite{cha-cou:j:cc}
in the context of ordinal elections and was adapted to the approval
setting by Procaccia et
al.~\cite{pro-ros-zoh:j:proportional-representation} and Betzler et
al.~\cite{bet-sli-uhl:j:mon-cc}.  Its purpose is to find diverse
committees, so that as many voters as possible feel represented by at
least one of the committee members.  The PAV rule was introduced by
Thiele~\cite{thi:j:pav} and its more elaborate scoring system is
designed to ensure proportional representation of the
voters~\cite{azi-bri-con-elk-fre-wal:j:justified-representation,bri-las-sko:j:apportionment}.

Both CC and PAV are $\np$-hard to
compute~\cite{pro-ros-zoh:j:proportional-representation,sko-fal-lan:j:collective,azi-gas-gud-mac-mat-wal:c:approval-multiwinner}
and we are mostly interested in the rules defined by their greedy
approximation algorithms. These algorithms run as follows (let $E = (C,V)$ be the
input election, $k$ be the committee size, and $\lambda$ be the OWA
function used):
\begin{enumerate}
\item[] We start with an empty committee $W = \emptyset$ and perform
  $k$ iterations, where in each iteration we extend $W$ with a single
  candidate $c$ that %
  maximizes the
  value
  $\lambda\hbox{-}\score_E(W \cup \{c\}) -
  \lambda\hbox{-}\score_E(W)$.  If several candidates satisfy this
  condition then we break the tie according to a given tie-breaking
  order. We output $W$ as the unique winning committee.
\end{enumerate}
We refer to the incarnations of this algorithm for $\lambda_\cc$ and
$\lambda_\pav$ as GreedyCC and GreedyPAV, respectively.
When analyzing an $i$-th iteration of these algorithms, for each
candidate $c$ we refer to the value
$\lambda\hbox{-}\score_E(W \cup \{c\}) - \lambda\hbox{-}\score_E(W)$ as the
score of $c$.  For GreedyCC, we imagine that as soon as a candidate is
included in the committee, all the voters that approve him or her are
removed (indeed, these voters would not contribute positive score to any
further candidates).

We are also interested in the Phragm{\'e}n rule (or, more specifically,
in the Phragm{\'e}n's sequential rule, but we use the shorter name in this paper). The Phragm{\'e}n rule proceeds according to the
following scheme ($E = (C,V)$ is the input election and~$k$ is the
committee size):
\begin{enumerate}
\item[] Initially, we have committee $W = \emptyset$. The voters start
  with no money, but they receive it at a constant rate (so, at each
  time point $t \in \mathbb{R}$, $t \geq 0$, each voter has in total
  received money of value $t$). At every time point for which there is
  a candidate $c$ not included in $W$ who is approved by voters that
  jointly have one unit of money, this candidate is ``purchased.''
  That is, candidate~$c$ is added to $W$ and the voters that approve
  him or her have all their money reset to $0$ (i.e., they pay for
  $c$). If several candidates can be purchased at the same time, we
  consider them in a given tie-breaking order. The process continues
  until $W$ reaches size $k$ or all the remaining candidates have
  approval score zero (in which case we extend $W$ according to the
  tie-breaking order). We output $W$ as the unique winning committee.
\end{enumerate}

Similarly to PAV, Phragm\'en provides committees that ensure
proportional representation of the
voters~\cite{san-elk-lac-fer-fis-bas-sko:c:pjr}.  For a detailed
discussion of these rules (including an alternative definition of
Phragm{\'e}n), we point the reader to the survey of Lackner and
Skowron~\cite{lac-sko:t:approva-multiwinner-survey}.  Faliszewski et
al.~\cite{fal-sko-sli-tal:b:multiwinner-voting} offer a general
overview of multiwinner voting.
Note that GreedyCC, GreedyPAV, and Phragm{\'e}n are resolute.

\subsection{Robustness of Multiwinner Voting Rules}

We use the robustness framework  introduced by Bredereck
et al.~\cite{bre-fal-kac-nie-sko-tal:c:robustness} for the ordinal
setting and adapted to the approval one by Gawron and
Faliszewski~\cite{gaw-fal:c:approval-robustness}.
In particular, we consider the \textsc{Add} and \textsc{Remove}
operations, where the former means adding a single approval to some
vote and the latter means removing a single approval from a vote. Let
us fix committee size $k$. For an operation \textsc{Op}
$\in \{\text{\textsc{Add}}, \text{\textsc{Remove}}\}$, we say that a
multiwinner voting rule $\calR$ is $\ell$-\textsc{Op}-robust (or, that
its \textsc{Op}-robustness level is~$\ell$) if~$\ell$ is the smallest
integer such that for every election $E = (C,V)$, where
$|C| \geq 2k$,\footnote{This is mostly a technical assumption, to ensure that there are enough candidates
so that all members of a committee can be replaced with non-members.}  and
every election $E'$ obtained from $E$ with a single operation of type
\textsc{Op}, the following holds:
\begin{enumerate}
\item[] For each committee $W \in \calR(E,k)$ there is a committee
  $W' \in \calR(E',k)$ such that $|W \cap W'| \geq k - \ell$ (i.e.,
  for every winning committee of $E$ there is a winning committee of
  $E'$ that differs in at most $\ell$ candidates).
\end{enumerate}
Intuitively, if a rule is $1$-\textsc{Add}-robust then adding a single
approval in an election held according to this rule may, at most, lead
to replacing a single member of the winning committee. On the other
hand, if a rule is $k$-\textsc{Add}-robust, then adding a single
approval sometimes leads to replacing the whole committee.  Gawron and
Faliszewski~\cite{gaw-fal:c:approval-robustness} have shown that AV is
$1$-\{\textsc{Add},\textsc{Remove}\}-robust, whereas both CC and PAV
are $k$-\{\textsc{Add},\textsc{Remove}\}-robust (they also considered
the \textsc{Swap} operation, which means moving an approval from one
candidate to another within a vote, and obtained analogous results for
it).

Following Bredereck et al.~\cite{bre-fal-kac-nie-sko-tal:c:robustness}
and Gawron and Faliszewski~\cite{gaw-fal:c:approval-robustness}, we
also study the \textsc{Robustness-Radius} problem. Intuitively, in
this problem we are interested in the smallest number of operations
required to change the election result. The more are necessary, the
more robust is the result on the given election.

\begin{definition}
  Let $\calR$ be a multiwinner voting rule and let \textsc{Op} be
  either \textsc{Add} or \textsc{Remove}. In the
  $\calR$-\textsc{Op}-\textsc{Robustness-Radius} problem we are given
  an election~$E$, a committee size $k$, and a nonnegative integer $B$
  (referred to as the budget). We ask if it is possible to perform up
  to $B$ operations of type \textsc{Op} so that for the resulting
  election $E'$ it holds that $\calR(E,k) \neq \calR(E',k)$.
\end{definition}

\section{Robustness Level}

The results of Bredereck et
al.~\cite{bre-fal-kac-nie-sko-tal:c:robustness} and Gawron and
Faliszewski~\cite{gaw-fal:c:approval-robustness} give some intuitions
regarding robustness levels that we may expect from multiwinner
rules.
On the one hand, simple, polynomial-time computable rules that
focus on individual excellence tend to have robustness levels equal to~$1$ (this includes, e.g., AV in the approval setting, and a number of
rules in the ordinal one).
Indeed,  Bredereck et
al.~\cite[Theorem~6]{bre-fal-kac-nie-sko-tal:c:robustness} have shown that  if a
rule selects a committee with the highest score, this score is easily
computable, and the rule's robustness level is bounded by a constant,
then some winning committee can be computed in polynomial time.
On the other hand, more involved rules that focus on proportionality
or diversity---in particular those $\np$-hard to compute---tend to
have robustness levels equal to the committee size.
However, regarding rules that form the committee sequentially, so far
there was only one data point: Bredereck et
al.~\cite{bre-fal-kac-nie-sko-tal:c:robustness} have shown that single
transferable vote (STV; a well-known rule for the ordinal setting)
has robustness levels equal to the committee size. We provide
three more such examples by showing that GreedyCC, GreedyPAV, and
Phragm{\'e}n also have robustness levels equal to the committee size.

First, we consider the relationship between \textsc{Add-Robustness}
and \textsc{Remove-Robustness} for resolute rules.

\begin{proposition}\label{prop:k-robustness}
  Let $\calR$ be a resolute multiwinner voting rule, and let $\ell$ be
  a positive integer.  $\calR$ is $\ell$-\textsc{Add}-robust if and
  only if it is $\ell$-\textsc{Remove}-robust.
\end{proposition}
\begin{proof}
  Let us fix committee size $k$ and a resolute multiwinner rule
  $\calR$. Further, assume that $\calR$ is $\ell$-\textsc{Add}-robust
  for some integer $\ell$. We will show that it also is
  $\ell$-\textsc{Remove}-robust. To see this, let us consider two
  elections, $E$ and $E'$, where $E'$ is obtained from $E$ by removing
  an approval and both elections contain at least $2k$
  candidates. Let~$W$ be the unique winning committee in $\calR(E,k)$
  and $W'$ be the unique winning committee in $\calR(E',k)$. Since
  $\calR$ is $\ell$-\textsc{Add}-robust, we know that $W$ and $W'$
  differ by at most $\ell$ candidates (it suffices to apply the
  defintion of $\ell$-\textsc{Add}-robustness, but with the roles of
  $E$ and $E'$ reversed). Similarly, we see that there are two such
  elections whose winning committees differ by exactly $\ell$
  candidates.
  By applying analogous reasoning, we see that if $\calR$ is
  $\ell$-\textsc{Remove}-robust then it is also
  $\ell$-\textsc{Add}-robust.
\end{proof}

Next we show that GreedyCC, GreedyPAV, and Phragm{\'e}n are
$k$-$\{$\textsc{Add, Remove}$\}$-robust.

\begin{theorem}
  Let $k$ be the committee size.  For each multiwinner rule $\calR$ in
  $\{\text{GreedyCC},\text{GreedyPAV},\text{Phragm{\'e}n}\}$, $\calR$ is both
  $k$-\textsc{Add}-robust and $k$-\textsc{Remove}-robust.\end{theorem}
\begin{proof}
  Let us fix committee size $k$.  Since our rules are resolute, by
  Proposition~\ref{prop:k-robustness} it suffices to show their
  $k$-\textsc{Add}-robustness.  To this end, we will form two
  elections, $E = (C,V)$ and $E' = (C,V')$, where $E'$ can be obtained
  from $E$ by adding a single approval, such that for each of our
  rules the unique winning committee for $E$ is disjoint from the one
  for $E'$.

  We let the candidate set be $C = A \cup B$, where
  $A = \{a_1, \ldots, a_k\}$ and $B = \{b_1, \ldots, b_k\}$, and
  we set the tie-breaking order to be:
  \[
    a_1 \pref \cdots \pref a_k \pref b_1 \pref \cdots \pref b_k.
  \]
  The voter collection of election $E$ is as follows:
  \begin{enumerate}
  \item We have $k-1$ voters approving $\{a_1, b_1\}$.
  \item For each $i \in \{2,\ldots,k\}$ we have a single voter approving $\{a_1, b_i\}$.
  \item For each $i \in \{2,\ldots,k\}$ we have a single voter approving $\{a_i, b_1\}$.
  \item For each $i \in \{2,\ldots,k\}$ we have $2k-3$ voters approving $\{a_i, b_i\}$.
  \item We have voter $v_0$ with empty approval set.
  \end{enumerate}
  As the reader can verify, every candidate from $C$ is approved by
  exactly $2(k-1)$ voters. Election $E'$ is defined in the same way,
  except that voter $v_0$ approves~$b_1$. For each of our rules we
  will show that committee $A$ wins in election $E$ and committee $B$
  wins in election $E'$.

  We first consider GreedyCC and election $E$. At the beginning of the
  first iteration, each candidate has score $2(k-1)$ and, due to the
  tie-breaking order, GreedyCC chooses~$a_1$. As a consequence, in the
  second iteration the score of~$b_1$ decreases by $k-1$ points and
  the scores of candidates $b_2, \ldots, b_k$ decrease by~$1$~point
  each.  As there is no voter who approves two different candidates
  from~$A$, each of the candidates in $\{a_2, \ldots, a_k\}$ still has
  $2(k-1)$ points. Hence~$a_2$ is selected. The same reasoning applies
  to the following $k-2$ iterations, during which the remaining
  members of $A$ are chosen.

  On the other hand, for election $E'$ GreedyCC outputs committee $B$.
  To see this, note that in the first iteration $b_1$ has score higher
  by one point than every other candidate and, so, is selected
  irrespective of the tie-breaking order.  Then the scores of
  candidates in $A$ decrease below $2(k-1)$, but the scores of
  candidates in $\{b_2, \ldots, b_k\}$ remain equal to $2(k-1)$. Hence
  these candidates are selected in the following iterations.
  Since GreedyCC outputs committee $A$ for election $E$ and committee
  $B$ for election $E'$, we see that GreedyCC is
  $k$-\textsc{Add}-robust.

  The case of GreedyPAV is analogous to that of GreedyCC. Indeed, the
  only difference between the operation of GreedyPAV and GreedyCC on
  elections $E$ and $E'$ is that when under GreedyCC the score of some
  candidate drops by $x$, the score of the same candidate drops by
  $\nicefrac{x}{2}$ under GreedyPAV (naturally, this is not a general
  feature of these rules, but one that is specific to our two
  elections). As a consequence, both rules choose the same committees
  for $E$ and $E'$ and, so, GreedyPAV is $k$-\textsc{Add}-robust.

  Finally, we consider the Phragm{\'e}n rule.  Since in election $E$
  each candidate is approved by $2(k-1)$ voters, the first moment when
  a group of voters can purchase a candidate is
  $\nicefrac{1}{2(k-1)}$. Due to the tie-breaking order, they first
  buy $a_1$, followed by $a_2$ and all the other members of $A$ (as no
  two members of $A$ are approved by the same voter, for each member
  of $A$ there is a group of voters with sufficient funding). Thus the
  rule outputs committee $A$. In election $E'$, candidate $b_1$ has
  $2k-1$ approvals and is purchased at time $\nicefrac{1}{2k-1}$. As a
  consequence, all voters who approve $b_1$ have their budgets reset
  to $0$. The next time when there is a group of voters that can
  purchase a candidate is $\nicefrac{1}{2(k-1)}$.  One can verify that
  at this point for each candidate in $\{b_2, \ldots, b_k\}$ there is
  a disjoint group of voters that has a unit of money, whereas there
  is no such group for any member of~$A$. Hence, Phragm{\'e}n outputs
  committee $B$. As in the previous two cases, this means that
  Phragm{\'e}n is $k$-\textsc{Add}-robust.
\end{proof}

The reader may worry that the above results hold due to the fact that
our rules are resolute, but this is not the case. For example, if one
used parallel-universes tie-breaking (where a rule outputs all the
committees that could win for some way of resolving the internal ties)
then the result would still hold. For example, for GreedyCC it would
suffice to add one more voter approving both $a_1$ and $b_1$ to
elections $E$ and $E'$. Then, GreedyCC with parallel-universes
tie-breaking would output both $A$ and $B$ as the winning committees
for $E$, but for $E'$ it would output only $B$. This would show its
$k$-\textsc{Add}-robustness (from the point of view of committee $A$).

\section{Robustness Radius: Complexity Results}

In this section we show that the \textsc{Robustness-Radius} problem is
$\np$-complete for each of our rules, for both adding and removing
approvals. We observe that for each of our rules and operation type,
the respective \textsc{Robustness-Radius} problem is clearly in
$\np$. Indeed, it suffices to nondeterministically guess which
approvals to add/remove, compute the winning committees before and
after the change (since our rules are resolute, in each case there is
exactly one), and verify that they are different. Hence, in our proofs
we will focus on showing $\np$-hardness. To this end, we give
reductions from the following variant of the \textsc{X3C} problem (it
is well known that this variant of the problem remains
$\np$-complete~\cite{gon:j:x3c}; note that in the standard variant of
X3C one does not assume that each member of $U$ belongs to exactly
three sets).

\begin{definition}
  An instance of \textsc{RX3C} consists of a universe set $U = \{u_1,$
  $\ldots, u_{3n}\}$ and a family $\calS = \{S_1, \ldots, S_{3n}\}$ of
  three-element subsets of $U$, such that each member of $U$ belongs
  to exactly three sets from~$\calS$. We ask if there is a collection
  of $n$ sets from $\calS$ whose union is $U$ (i.e., we ask if there
  is an exact cover of $U$).
\end{definition}

All our reductions follow the same general scheme: Given an instance
of \textsc{RX3C} we form an election where the sets are the candidates
and the voters encode their content. Additionally, we also have
candidates $p$ and $d$. Irrespective which operations we perform
(within a given budget), all the set candidates are always selected, but
by performing appropriate actions we control the order in which this
happens. If the order corresponds to finding an exact cover, then
additionally candidate $p$ is selected. Otherwise, our rules select
$d$.

We first focus on adding approvals and then argue why our proofs adapt
to the case of removing approvals.

\begin{theorem}\label{thm:greedy-cc}
  GreedyCC-\textsc{Add-Robustness-Radius} is $\np$-complete.
\end{theorem}
\begin{proof}
  We give a reduction form the \textsc{RX3C} problem. Our input
  consists of the universe set $U = \{u_1, \ldots, u_{3n}\}$ and
  family $\calS = \{S_1, \ldots, S_{3n}\}$ of three-element subsets of
  $U$. We know that each member of $U$ belongs to three sets
  from~$\calS$.
  We introduce two integers, $T = 10n^5$ and $t = 10n^3$ and we
  interpret both as large numbers, with $T$ being significantly larger
  than $t$.  We form an election $E = (C,V)$ with candidate set
  $C = \{S_1, \ldots, S_{3n}\} \cup \{p,d\}$, and with the following
  voters:
  \begin{enumerate}
  \item For each $S_i \in \calS$, there are $T$ voters that approve
    candidate $S_i$.
  \item For each two sets $S_i$ and $S_{j}$, there are $T$ voters
    that approve candidates $S_i$ and $S_{j}$.
  \item There are $2nT+4nt$ voters that approve $p$ and $d$.
  \item For each $u_\ell \in U$, there are $t$ voters that approve $d$
    and those candidates $S_i$ that correspond to the sets containing
    $u_\ell$.
  \item There are $n$ voters who do not approve any candidates.
  \end{enumerate}
  The committee size is $k = 3n+1$ and the budget is $B = n$.  We
  assume that the tie-breaking order among the candidates is:
  \[
    S_1 \pref S_2 \pref \cdots \pref S_{3n} \pref p \pref d.
  \]
  Prior to any bribery, each candidate $S_i$ is approved by $3nT+3t$
  voters, $p$ is approved by $2nT+4nt$ voters, and $d$ is approved by
  $2nT+7nt$ voters.

  Let us now consider how GreedyCC operates on this election. Prior to
  the first iteration, all the set candidates have the same score,
  much higher than that of $p$ and $d$. Due to the tie-breaking order,
  GreedyCC chooses $S_1$. As a consequence, all the voters that
  approve $S_1$ are removed from consideration and the scores of all
  other set candidates decrease by $T$ (or by $T+t$ or $T+2t$, for the
  sets that included the same one or two elements of $U$ as
  $S_1$). GreedyCC acts analogously for the first $n$ iterations,
  during which it chooses a family $\calT$ of $n$ set elements (we
  will occasionally refer to $\calT$ as if it really contained the
  sets from~$\calS$, and not the corresponding candidates).
  
  After the first $n$ iterations, each of the remaining $2n$ set
  candidates either has $2nT$, $2nT+t$, $2nT+2t$, or $2nT+3t$ approvals
  (depending how many sets in $\calT$ have nonempty intersection with
  them). Let $x$ be the number of elements from $U$ that do not belong
  to any set in $\calT$.  Candidate $p$ is still approved by $2nT+4nt$
  voters, whereas $d$ is approved by $2nT+4nt + xt$ voters. Thus at
  this point there are two possibilities. Either $x = 0$ and, due to
  the priority order, GreedyCC selects $p$, or $x > 0$ and GreedyCC
  selects $d$. In either case, in the following $2n$ iterations it
  chooses the remaining $2n$ set candidates (because after the
  $n+1$-st iteration the score of that among $p$ and $d$ who remains
  drops to zero or nearly zero). If candidate $p$ is selected without
  any bribery, then it means that we can find a solution for the
  \textsc{RX3C} instance using a simple greedy algorithm. In this
  case, instead of outputting the just-described instance of
  GreedyCC-\textsc{Add-Robustness-Radius}, we output a fixed one, for
  which the answer is \emph{yes}. Otherwise, we know that without any
  bribery the winning committee is $\{S_1, \ldots, S_{3n},d\}$. We
  focus on this latter case.
  
  We claim that it is possible to ensure that the winning committee
  changes by adding at most $n$ approvals if and only if there is an
  exact cover of $U$ with~$n$ sets from $\calS$.  Indeed, if such a
  cover exists, then it suffices to add a single approval for each of
  the corresponding sets in the last group of voters (those that
  originally do not approve anyone). Then, by the same analysis as in
  the preceding paragraph, we can verify that the sets forming the
  cover are selected in the first $n$ iterations, followed by $p$,
  followed by all the other set candidates.
  
  For the other direction, let us assume that after adding some $t$
  approvals the winning committee has changed. One can verify that
  irrespective of which (up-to) $n$ approvals we add, in the first $n$
  iterations GreedyCC still chooses $n$ set candidates. Thus, at this
  point, the score of $p$ is at most $2nT+4nt+n$ and the score of $d$ is
  at least $2nT+4nt+xt-n$ (where $x$ is the number of elements from
  $U$ not covered by the chosen sets; we subtract $n$ to account for
  the fact that $n$ voters that originally approved $d$ got approvals
  for the candidates selected in the first $n$ interations). If at
  this point $d$ is selected, then in the following $2n$ iterations
  the remaining set candidates are chosen and the winning committee
  does not change. This means that $p$ is selected. However, this is
  only possible if $x = 0$, i.e., if the set candidates chosen in the
  first $n$ iterations correspond to an exact cover of $U$.
\end{proof}

A very similar proof also works for the case of GreedyPAV. The main
difference is that now including a candidate in a committee does not
allow us to forget about all the voters that approve him or her (the
proof is in the appendix).

\begin{theorem}\label{thm:greedy-pav}
  GreedyPAV-\textsc{Add-Robustness-Radius} is $\np$-complete.
\end{theorem}

The proof for the case of Phragm{\'e}n-\textsc{Add-Robustness-Radius}
is similar in spirit to the preceding two, but requires careful
calculation of the times when particular groups of voters can purchase
respective candidates.

\begin{theorem}\label{thm:phr}
  Phragm{\'e}n-\textsc{Add-Robustness-Radius} is $\np$-complete.
\end{theorem}
\begin{proof}
  We give a reduction from \textsc{RX3C}. As input, we get a universe
  set $U = \{u_1, \ldots, u_{3n}\}$ and a family
  $\calS = \{S_1, \ldots, S_{3n}\}$ of size-$3$ subsets of $U$. Each
  element of $U$ appears in exactly three sets from $\calS$.
  We ask if
  there is a collection of $n$ sets that form an exact cover of $U$.

  Our reduction proceeds as follows. First, we define two numbers,
  $T = 900n^{12}$ and $t = 30n^5$.  The intuition is that both numbers
  are very large, $T$ is significantly larger than $t^2$, and $t$ is
  divisible by $6n$ (the exact values of $T$ and $t$ are not crucial;
  we did not minimize them but, rather, used values that clearly work
  and simplify the reduction). We form an election $E = (C,V)$ with
  candidate set $C = \{S_1, \ldots, S_{3n}\} \cup \{p,d\}$ and the
  following voters:
  \begin{enumerate}
  \item For each $S_i \in \calS$, there are $T$ voters that approve
    candidate $S_i$. We refer to them as the $\calS$-voters.
  \item For each $u_\ell \in U$, there are $t^2$ voters that approve
    those candidates $S_i$ that correspond to the sets containing
    $u_\ell$. We refer to them as the universe voters.  For each
    $u_\ell \in U$, $\frac{t}{3n}$ of $u_\ell$'s universe voters
    additionally approve candiate $d$.  We refer to them as the
    $d$-universe voters.
    
  \item There are $T+3t^2-2t$ voters that approve both $p$ and $d$. We refer to
    them as the $p$/$d$-voters.
  \item There are $\frac{t}{6n}$ voters that approve
    $p$. We refer to them as the $p$ voters.
  \item There are $n$ voters who do not approve any candidate, and to
    whom we refer as the empty voters.
  \end{enumerate}
  The committee size is $k = 3n+1$ and the budget is $B = n$. The tie-breaking
  order is:
  \[
    S_1 \pref S_2 \pref \cdots \pref S_{3n} \pref d \pref p.
  \]
  In this election, each candidate $S_i$ is approved by exactly
  $T+3t^2$ voters, $d$ is approved by $(T+3t^2-2t) + t$ voters, and
  $p$ is approved by $(T+3t^2-2t) +\frac{t}{6n}$
  voters.

  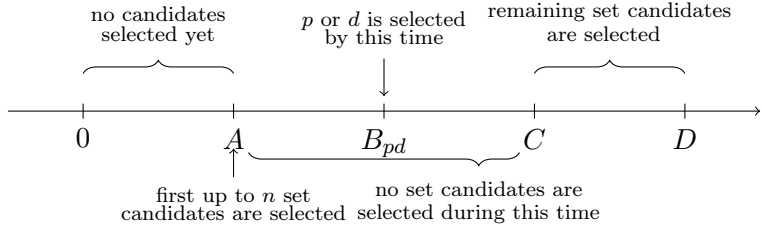
\begin{figure}[t]
    \centering
    \begin{tikzpicture}[scale=1]
      \draw[->] (0,0) -- (10,0);
      \foreach \x\l in {1/$0$,3/$A$,5/$B_{pd}$,7/$C$,9/$D$}{
        \draw (\x,0.1) -- (\x,-0.1) node[anchor=north] {\l};
        
      }
      \draw[decorate, decoration = {brace, amplitude=5pt}] (7,0.5) -- (9,0.5) node[pos=0.5,above=16pt,black]{{\scriptsize remaining set candidates}} node[pos=0.5,above=7pt,black]{{\scriptsize are selected\phantom{y}}};      

      \draw[decorate, decoration = {brace, amplitude=5pt}] (1,0.5) -- (3,0.5) node[pos=0.5,above=16pt,black]{{\scriptsize no candidates}} node[pos=0.5,above=7pt,black]{{\scriptsize selected yet}};      

      \draw[<-] (3,-0.5) -- (3,-0.88) node[anchor=north] {{\scriptsize first up to $n$ set}} node[below=7pt]{{\scriptsize candidates are selected}};
      \draw[<-] (5,0.2) -- (5,0.7)  node[anchor=south,above=7pt] {{\scriptsize $p$ or $d$ is selected}} node[anchor=south,above=0pt] {{\scriptsize by this time}};

      \draw[decorate, decoration = {brace, amplitude=5pt,mirror,aspect=0.85}] (3.2,-0.55) -- (6.8,-0.55) node[pos=0.85,below=7pt,black]{{\scriptsize no set candidates are}} node[pos=0.85,below=16pt,black]{{\scriptsize selected during this time}};     
      
    \end{tikzpicture}
    
    \caption{\label{fig:phr}Timeline for the Phragm{\'e}n rule acting
      on the election from Theorem~\ref{thm:phr}.}
  \end{figure}

  Let us consider how Phragm{\'e}n operates on this election (we
  encourage the reader to consult Figure~\ref{fig:phr} while reading
  the following text). First, we observe that when we reach time point
  $D = \frac{1}{T}$ then all the not-yet-selected set candidates (for
  whom there still is room in the committee) are selected. Indeed, at
  time $D$ the $\calS$-voters collect enough funds to buy them. On the
  other hand, the earliest point of time when some voters can afford
  to buy a candidate is $A = \frac{1}{T+3t^2}$. Specifically, at time
  $A$ set voters and universe voters jointly purchase up to $n$ set
  candidates (selected sequentially, using the tie-breaking order and
  taking into account that when some candidate is purchased then all
  his or her voters spend all their so-far collected money). Let us
  consider some candidate $S_i$ that was not selected at time point
  $A$. Since $S_i$ was not chosen at $A$, at least $t^2$ of the $3t^2$
  universe voters that approve $S_i$ paid for another candidate at
  time $A$. Thus the earliest time when voters approving $S_i$ might
  have enough money to purchase him or her is $C$, such that:
  \[
    \underbrace{C(T+2t^2)}_{\substack{\text{money earned by those voters} \\ \text{who did not spend it at time $A$}}} +
    \underbrace{(C-A)t^2}_{\substack{\text{money earned between times $C$ and $A$ by uni-} \\ \text{verse voters who paid for candidates at time $A$}}}
    = 1.
  \]
  Simple calculations show that $C = \frac{1+At^2}{T+3t^2}$. Noting
  that $A = \frac{1}{T+3t^2}$, we have that $C = A + A^2t^2$.
  However, prior to reaching time point $C$, either candidate $p$ or
  candidate $d$ is selected. Indeed, at time point
  $B_{pd} = \frac{1}{T+3t^2-2t}$ the $p$/$d$ voters alone would have
  enough money to buy one of their candidates:  We show that $B_{pd} <
  C$, or, equivalently,  that $\frac{1}{B_{pd}} > \frac{1}{C}$.  It
  holds that $\frac{1}{B_{pd}} = T+3t^2 - 2t$ and:
  \[
    \frac{1}{C} = \frac{1}{A+A^2t^2} = \frac{1}{A} \cdot \frac{1}{1+At^2}
    = \frac{T+3t^2}{1+\frac{t^2}{T+3t^2}}
    = \frac{(T+3t^2)^2}{T+4t^2}.
  \]
  By simple transformations, $\frac{1}{B_{pd}} > \frac{1}{C}$ is equivalent to:
  \[
    (T+3t^2 - 2t)(T+4t^2) > (T+3t^2)^2.
  \]
  The left-hand side of this inequality can be expressed as:   
  \[
    ((T+3t^2) - 2t)((T+3t^2) + t^2) = (T+3t^2)^2 +
    \underbrace{(t^2-2t)(T+3t^2) - 2t^3}_{\substack{\text{positive because $t^2-2t > 2t$} \\ \text{due to our assumptions} }},
  \]
  and, hence, our inequality holds. All in all, we have
  $A < B_{pd} < C < D$.

  It remains to consider which among $p$ and $d$ is selected.  If $p$
  were to be selected, then it would happen at time point
  $B_p = \frac{1}{T^2+3t^2-2t + \frac{t}{6n}}$. This is when the
  $p$/$d$- and $p$ voters would collect enough money to purchase $p$
  (assuming the former would not spend it on $d$ earlier).  Now, if at
  time $A$ fewer than $n$ set candidates were selected, then at least
  $\frac{t}{3n}$ of the $d$-universe voters would retain their money
  and, hence, $d$ would be selected no later than at time point
  $B_d = \frac{1}{T^2+3t^2-2t+\frac{t}{3n}} < B_p$.  On the other
  hand, if at time point $A$ exactly $n$ set candidates were selected
  (who, thus, would have to correspond to an exact cover of $U$) then all
  the $d$-universe voters would lose their money and voters who
  approve~$d$ would not have enough money to buy him or her before
  time $B_p$.  Indeed, in this case the money accumulated by voters
  approving $d$ would at time $B_p$ be:
  \[
    X = \underbrace{\frac{T+3t^2-2t}{T+3t^2-2t +
        \frac{t}{6n}}}_{\text{money of the $p$/$d$ voters}} +
    \underbrace{t \left( \frac{1}{T+3t^2-2t+\frac{t}{6n}} -
        \frac{1}{T+3t^2}\right)}_{\substack{\text{money collected by
          the $d$-universe} \\ \text{voters between time points $A$ and
          $B_p$}}}
  \]
  We claim that $X < 1$, which is equivalent to the following
  inequality (where we replace $T+3t^2$ with $M$; note that
  $M = \frac{1}{A}$):
  \begin{align*}
    \frac{M-t}{M-2t+\frac{t}{6n}} < 1 + \frac{t}{M} = \frac{M+t}{M} 
  \end{align*}
  By simple transformations, this inequality is equivalent to:
  \begin{align*}
  0 < \frac{Mt+t^2}{6n} - 2t^2, 
  \end{align*}
  which holds as $t > 6n$ and $M > 2t^2$.  To conclude, if at time
  point $A$ there are $n$ set candidates selected for the committee,
  then $p$ is selected for the committee as well.

  Finally, we observe that irrespective of which among $p$ and $d$ is
  selected for the committee, the voters that approve the other one do
  not collect enough money to buy him or her until time $D$.  Thus the
  winning committee either consists of all the set candidates and $d$,
  or of all the set candidates and $p$, where the latter happens
  exactly if at time $A$ candidates corresponding to an exact cover of
  $U$ are selected.\medskip

  If at point $A$ Phragm{\'e}n would choose candidates corresponding
  to an exact cover of $U$ then our reduction outputs a fixed
  yes-instance of Phragm{\'e}n-\textsc{Add-Robustness-Radius} (as we
  have just found that an exact cover exists). Otherwise we output the
  formed election with committee size $k = 3n+1$ and budget $B = n$.
  To see why this reduction is correct, we make the following three
  observations:
  \begin{enumerate}
  \item By adding $n$ approvals, we cannot significantly modify any of
    the time points $A$, $B_d$, $B_p$, $B_{pd}$, $C$, and $D$ from the preceding
    analysis, except that we can ensure which (up to) $n$ sets are
    first considered for inclusion in the committee just before time
    point $A$.
  \item If there is a collection of $n$ sets in $\calS$ that form an exact
    cover of $U$, then---by the above observation---we can ensure that these
    sets are selected just before time point $A$ (by adding one approval for each
    of them among $n$ distinct empty voters). Hence, if there is an exact cover
    then---by the preceding discussions---we can ensure that the winning committee changes
    (to consist of all the set candidates and $p$).
  \item If there is no exact cover of $U$, then no matter where we add
    (up to) $n$ approvals, candidate $d$ gets selected and, so, the
    winning committee does not change (in particular, even if we add
    $n$ approvals for $p$).
  \end{enumerate}
  Since the reduction clearly runs in polynomial time, the proof is complete.
\end{proof}

It remains to argue that \textsc{Remove-Robustness-Radius} also is
$\np$-complete for each of our rules. This, however, is easy to
see. In each of the three proofs above, we had budget $B = n$ and $n$
voters with empty approval sets.  We were using these $n$ voters to add
a single approval for each of the $n$ sets forming an exact cover,
leading to the selection of $p$ instead of $d$. For the case of
removing approvals, it suffices to replace the $n$ empty voters with
$3n$ ones, such that each set candidate is approved by exactly one of
them, and to set the budget to $B = 2n$. Now we can achieve the same
result as before by deleting approvals for those set candidates that
do not form an exact cover. Hence the following holds.

\begin{corollary}
  Let $\calR$ be one of GreedyCC, GreedyPAV and
  Phragm{\'e}n. $\calR$-\textsc{Remove-Robustness-Radius} is
  $\np$-complete.
\end{corollary}

\section{Robustness to Random Noise: Experimental Results}

Let us now move on to an experimental analysis of our rules'
robustness to random noise. The main idea of the experiment is as
follows: First, we generate a number of elections from a given
distribution and, for each of them, we compute its winning
committee. Then, we perform a given number of random operations, such
as adding or removing approvals (specified via a \emph{perturbation
  level}, described below), and we compute the proportion of elections
that change their outcome and the average number of committee members
that get replaced. We do so for each of our rules (including AV), for
several distributions, and for a range of perturbation levels. Our
main observation is that the results for AV, PAV, and Phragm{\'e}n are
quite similar to each other, but those for CC stand out. Further, the
results may quite strongly depend on the distribution of votes. Below
we describe our setup and present the results in more detail.

\paragraph{Generating Elections.}

To generate synthetic elections, we use the \emph{resampling} model
recently introduced by
Szufa et al.~\cite{szu-fal-jan-lac-sli-sor-tal:c:sample-approval}. This model is
parameterized by two numbers, $p, \phi \in [0,1]$, and to generate an
election with candidate set $C = \{c_1, \ldots, c_m\}$ and voter
collection $V = (v_1, \ldots, v_n)$ it proceeds as follows: First, we
choose a central approval set $A$ of $\lfloor p \cdot m \rfloor$
candidates from $C$ (uniformly at random from all subset of $C$ of
this cardinality).  Then, for each voter $v_i$ we set his or her
initial approval set $A(v_i)$ to be equal to $A$.  Finally, for each
voter $v_i$ and each candidate $c_j$, with probability $\phi$ we
remove $c_j$ from the approval set of $v_i$ (if it were there) and,
then, we let $v_i$ approve $c_j$ with probability~$p$. In other words,
initially all voters have the central approval set, but for each
candidate we resample its approval with probability $\phi$.
For example, for $\phi = 0$ each voter has identical approval set,
which includes $\lfloor p \cdot m\rfloor$ candidates, whereas for
$\phi = 1$ each voter approves each candidate independently, with
probability $p$. The closer $\phi$ is to $0$, the more similar are the
votes, and the closer it is to $1$, the more diverse they are.

\paragraph{Perturbation Levels.}
Given an election $E = (C,V)$, a perturbation level $\ell \in [0,1]$
specifies how many operations of adding or removing approvals we are
supposed to perform.  For the \textsc{Add} operation, perturbation
level $\ell$ means that we add an~$\ell$ fraction of all the approvals
not appearing in the election, chosen uniformly at random. (In our
election~$E$, there are $X = \sum_{v \in V}|A(v)|$ approvals in total,
but if each voter approved each candidate then there would be
$|C|\cdot|V|$ approvals. Thus the number of not appearing approvals is
$|C|\cdot|V| - X$.)  For the \textsc{Remove} operation, perturbation
level $\ell$ means removing an $\ell$ fraction of the approvals in the
election, chosen uniformly at random.

\begin{figure}
   \centering
      \includegraphics[width=.85\textwidth]{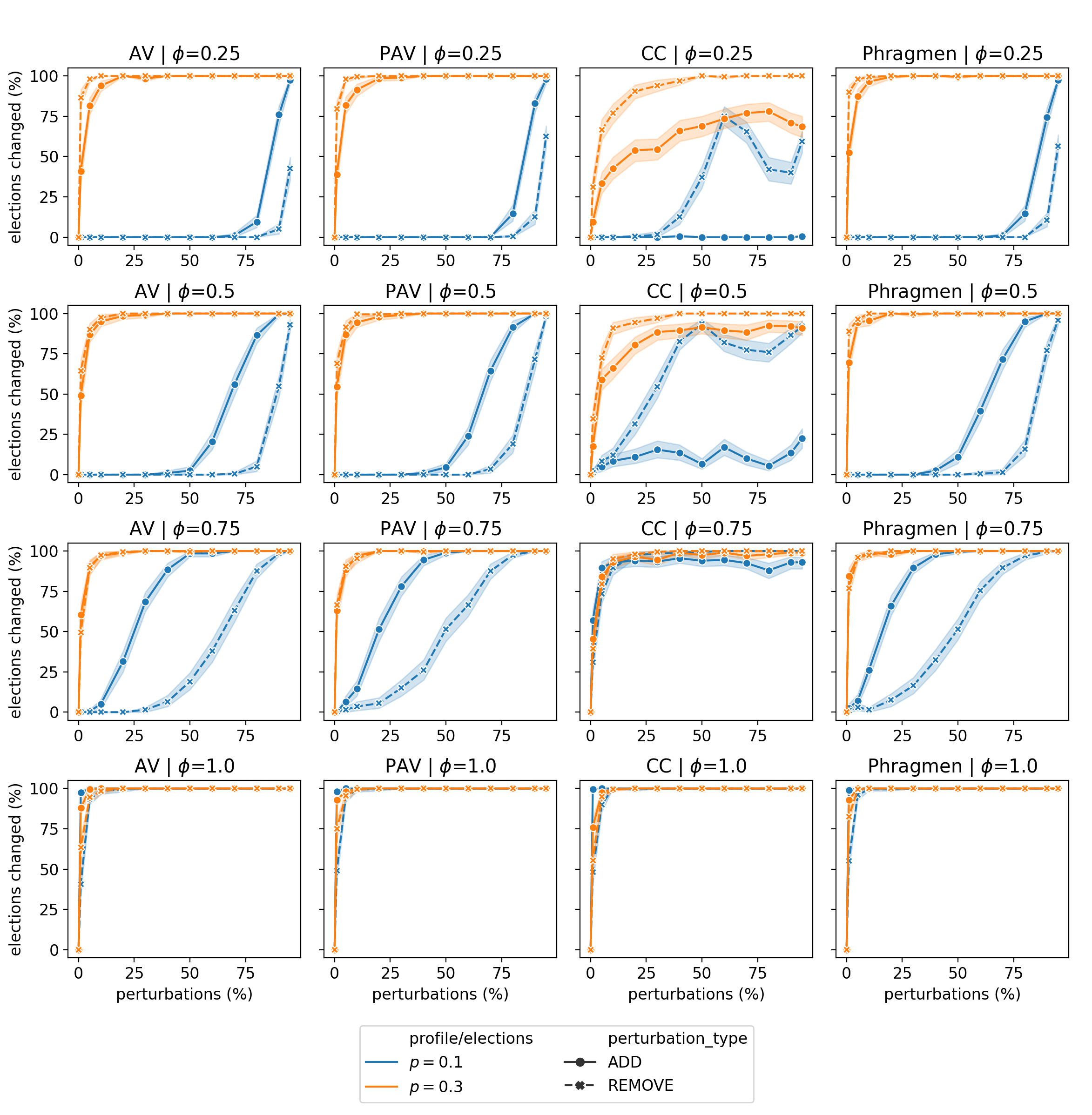}
      \caption{Probabilities of changing elections results for the
        resampling model with $p=0.1$ (blue lines) and $p=0.3$ (orange
        lines) for different rules (columns of the plot) and different
        values of $\phi$ (rows of the plot) and different perturbation
        levels ($x$ axis). Each data point corresponds to $200$
        elections with $100$ candidates, $100$ voters, and committee
        size $10$. Wide light blue and light orange lines represent
        standard deviation.}
      \label{fig:resampling:changed}
\end{figure}

\begin{figure}
   \centering
      \includegraphics[width=.85\textwidth]{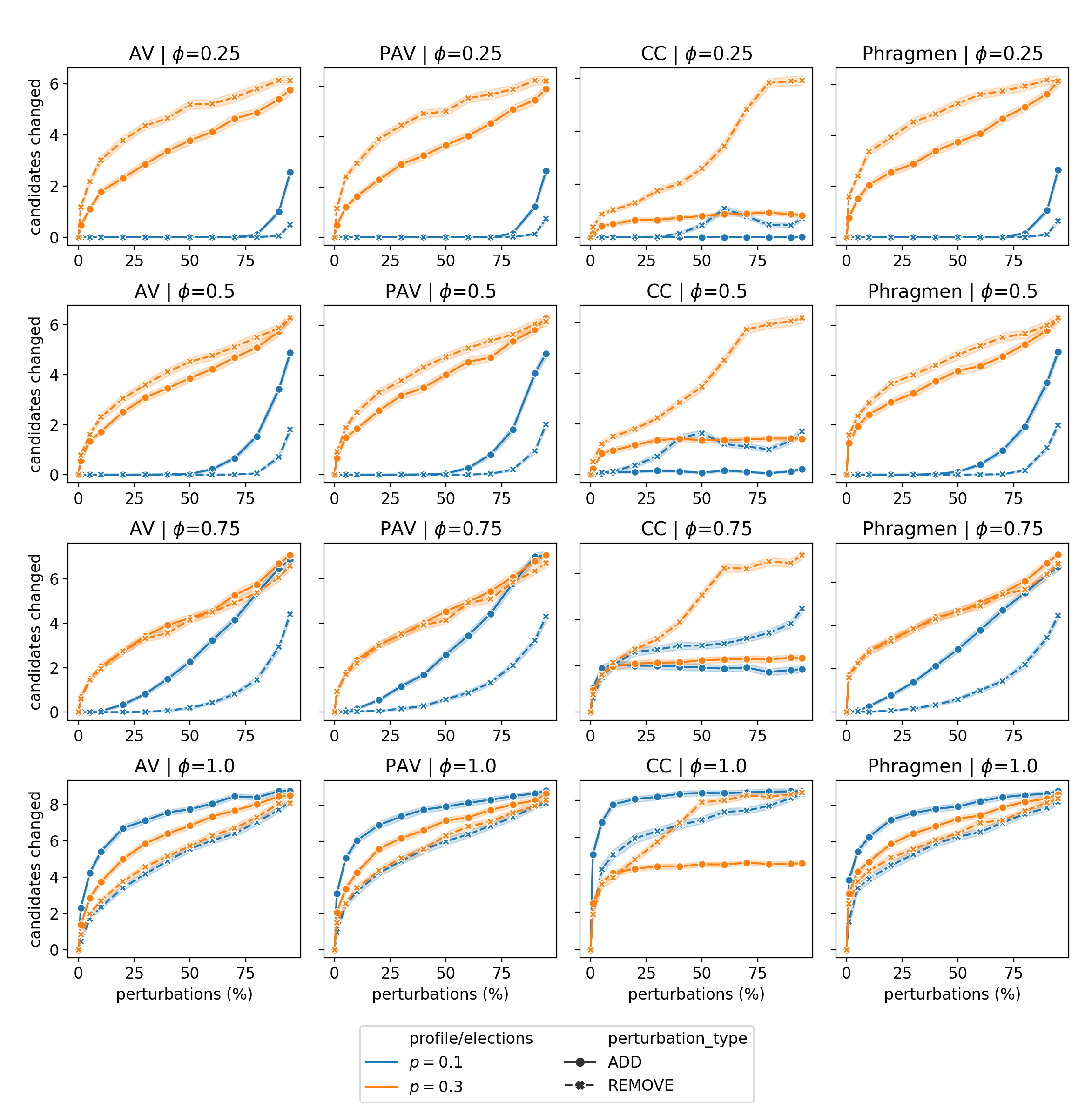}
      \caption{Average number of replaced committee members (the setup of the plot is analogous
        to the one in Figure~\ref{fig:resampling:changed}.}
      \label{fig:resampling:numcandidates}
\end{figure}

\paragraph{Performing the Experiment.}
To perform our experiment for a given multiwinner rule $\calR$, we
consider values of $p \in \{0.1, 0.3\}$, values of
$\phi \in \{0.25, 0.5, 0.75, 1\}$, perturbation levels $\ell$ between
$0$ and $0.95$, with a step of $0.05$ (but also including perturbation
level $0.01$), and operations \textsc{Op} $\in \{$\textsc{Add},
\textsc{Remove}$\}$. For each combination of these parameters we
generate $200$ elections with $100$ candidates and $100$ voters from
the resampling model with parameters $p$ and $\phi$. For each of these
elections we compute its $\calR$ winning committee of size $k = 10$,
apply operations \textsc{Op} as specified by the perturbation level,
and compute the winning committee of the resulting election (of the
same size). We report the fraction of elections for which the two
committees differ and the average number of candidates by which they
differ. We show the results in Figures~\ref{fig:resampling:changed}
and~\ref{fig:resampling:numcandidates}.

\paragraph{Analysis.}
The results in Figures~\ref{fig:resampling:changed}
and~\ref{fig:resampling:numcandidates} show several interesting
patterns. Most strikingly, the results for AV, PAV, and Phragm{\'e}n
are very similar to each other (to the point that it is often quite
difficult to distinguish respective plots), whereas those for CC stand
out sharply. This suggests that the nature of choosing diverse
committees, as done by CC, is quite different from that of choosing
individually excelent ones (as done by AV) or proportional ones (as
done by PAV and Phragm{\'e}n).

Second observation is that it is much easier to affect the results of
elections where the votes approve, on average, $p=0.3$ fraction of the
candidates (orange lines in Figures~\ref{fig:resampling:changed}
and~\ref{fig:resampling:numcandidates}) than those where they
approve, on average, $p=0.1$ fraction of them (blue lines in
Figures~\ref{fig:resampling:changed}
and~\ref{fig:resampling:numcandidates}). This is somewhat
counterintuitive. For example, in AV one would expect that with fewer
approvals in total it would be easier to push some non-winning
candidate into the committee by, say, adding approvals, because the
bar for entering the committee should be low. On the other hand, the
added approvals come from a wider set of possibilities.

The next observation is that the higher the value of $\phi$, the
easier it is to affect the output committees. This is intuitive as for
small values of $\phi$ the votes are highly correlated, whereas for
larger $\phi$ the votes are more random and more fragile to
adding noise.

\section{Summary}

We have complemented the results of Bredereck et
al.~\cite{bre-fal-kac-nie-sko-tal:c:robustness} and Gawron and
Faliszewski~\cite{gaw-fal:c:approval-robustness} by considering the
robustness of GreedyCC, GreedyPAV, and Phragm{\'e}n. We have found
that their robustness levels are equal to the committee size (which
means that even a minimal change to the votes can lead to completely
replacing the winning committee), that the problems of deciding if
modifying their input to a certain extent may change their outcomes
are $\np$-complete, and we have observed how these rules react to
random noise.

\section*{Acknowledgments}
This project has received funding from the European 
    Research Council (ERC) under the European Union’s Horizon 2020 
    research and innovation programme (grant agreement No 101002854).
    
\noindent \includegraphics[width=3cm]{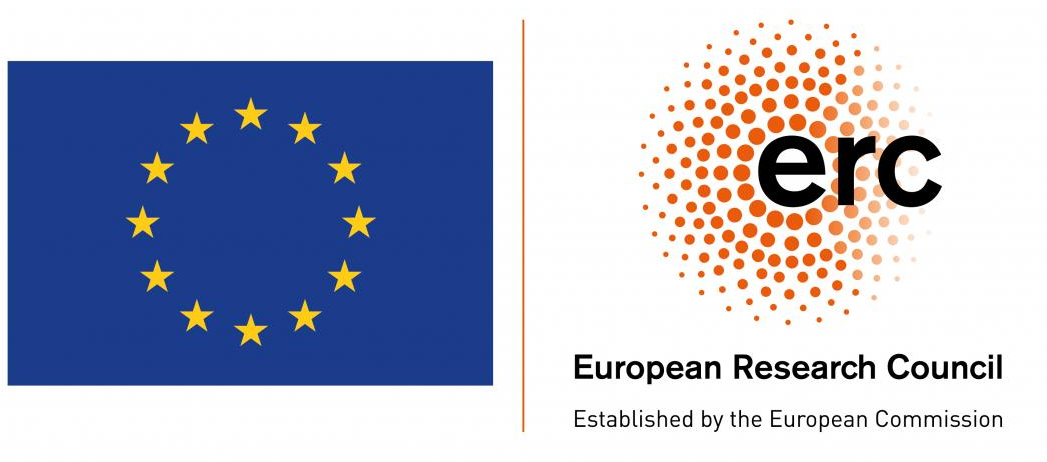}

\bibliography{bib-robust}

\appendix

\section{Proof of Theorem~\ref{thm:greedy-pav}}

\begin{proof}
  The proof is analogous to that of Theorem~\ref{thm:greedy-cc} and we
  only modify some details of the argument. We reduce from
  \textsc{RX3C} and the input instance consists of the universe set
  $U = \{u_1, \ldots, u_{3n}\}$ and family
  $\calS = \{S_1, \ldots, S_{3n}\}$ of three-element subsets of
  $U$. Each member of $U$ belongs to exactly three sets from~$\calS$.

  We have two integers, $T = 10n^5$ and $t = 10n^3$, interpreted as
  two large numbers, with $T$ significantly larger than $t$.  We form
  an election $E = (C,V)$ with candidate set
  $C = \{S_1, \ldots, S_{3n}\} \cup \{p,d\}$, and with the following
  voters:
  \begin{enumerate}
  \item For each $S_i \in \calS$, there are $T$ voters that approve
    candidate $S_i$.
  \item For each two sets $S_i$ and $S_{j}$, there are $T$ voters
    that approve candidates $S_i$ and $S_{j}$.
  \item There are $2nT+ 0.5nT+4nt$ voters that approve $p$ and $d$.
  \item For each $u_\ell \in U$, there are $t$ voters that approve $d$
    and those candidates $S_i$ that correspond to the sets containing
    $u_\ell$.
  \item There are $1.5nt$ voters who approve $p$.
  \item There are $n$ voters who do not approve any candidates.
  \end{enumerate}
  The committee size is $k = 3n+1$ and the budget is $B = n$.  We
  assume that the tie-breaking order among the candidates is:
  \[
    S_1 \pref S_2 \pref \cdots \pref S_{3n} \pref p \pref d.
  \]
  Prior to adding any approvals and running the rule, each candidate $S_i$ has score
  $3nT+3t$, $p$ has score $2nT + 0.5nT +4nt + 1.5nt$, and $d$ has score $2nT+0.5nT+4nt + 3nt$.

  Let us now consider how GreedyPAV operates on this election. Prior
  to the first iteration, all the set candidates have the same score,
  much higher than that of $p$ and $d$. Due to the tie-breaking order,
  GreedyPAV chooses $S_1$. As a consequence, the scores of all other
  set candidates decrease to $(3n-1)T + 0.5T$ plus some number of
  points from the fourth group of voters (but this part of the score
  is much smaller than that from the first two groups of
  voters). GreedyPAV acts analogously during the first $n$ iterations
  and it chooses a family $\calT$ of $n$ set elements.
  
  After the first $n$ iterations, each of the remaining $2n$ set
  candidates has $2nT + 0.5nT$ points from the first two groups of
  voters and at most $3nt$ points from the fourth group (in fact less,
  but this bound suffices). On the other hand, both $p$ and $d$ have
  at least $2nT + 0.5nT + 4nt$ points and, so, in the next
  iteration the algorithm chooses one of them.
  Specifically, $p$ has score $2nT+0.5nT+4nt+1.5nt$ and $d$ has score
  $2nT+0.5nT+4nt + 1.5nt + x$, where the value of $x$ is as
  follows. If~$\calT$ corresponds to an exact cover of $U$, then $x$
  is $0$ because for each set candidate that GreedyPAV adds during the
  first $n$ iteration, $d$ loses $1.5t$ points from the fourth group
  of voters.  However, if $\calT$ does not correspond to an exact
  cover of $U$, then for at least one added set candidate the score of
  $d$ does not decrease by $1.5t$ but by at most
  $t+(\nicefrac{1}{2}-\nicefrac{1}{3})t$. Thus $x$ is at least
  $\frac{t}{3}$. So, if $\calT$ corresponds to an exact cover of $U$
  then $p$ is selected and otherwise $d$ is.  In either case, the
  score of the unselected one drops by at least $1.25nT$, which means
  that in the following $2n$ iterations the remaining set candidates
  are selected (because each of them has score at least $1.5nT$).
  
  If candidate $p$ is selected without adding any approvals, then it
  means that we can find a solution for the \textsc{RX3C} instnace
  using a simple greedy algorithm. In this case, instead of outputting
  the just-described instance of
  GreedyPAV-\textsc{Add-Robustness-Radius}, we output a fixed one, for
  which the answer is \emph{yes}. Otherwise, we know that without any
  bribery the winning committee is $\{S_1, \ldots, S_{3n},d\}$. We
  focus on this latter case.

  If there is an exact cover of $U$ by sets from $\calS$ then it
  suffices to add a single approval for each of the corresponding sets
  in the last group of voters (those that originally do not approve
  anyone). Then, by the same analysis as above, we can verify that the
  sets forming the cover are selected in the first $n$ iterations,
  followed by $p$, followed by all the other set candidates.
  
  For the other direction, let us note that no matter which $n$
  approvals we add, it is impossible to modify the general scenario
  that GreedyPAV follows: It first chooses $n$ set candidates, then
  either $p$ or $d$ is selected (where the former can happen only if
  the first $n$ set canidates correspond to an exact cover of $U$),
  and finally the remaining $2n$ set candidates are chosen.  This is
  so, because adding $n$ voters results in modifying each of the
  scores by a value between $-n$ and $n$ and such changes do not
  affect our analysis from the preceding paragraphs.  This completes
  the proof.
\end{proof}

\end{document}